\documentclass[conference]{IEEEtran}
\IEEEoverridecommandlockouts

\usepackage{cite}
\usepackage{amsmath,amssymb,amsfonts}
\usepackage[ruled,vlined]{algorithm2e}
\usepackage{graphicx}
\usepackage{textcomp}
\usepackage{xcolor}
\usepackage{float}
\usepackage{url}
\def\BibTeX{{\rm B\kern-.05em{\sc i\kern-.025em b}\kern-.08em
    T\kern-.1667em\lower.7ex\hbox{E}\kern-.125emX}}

\usepackage{amsmath,amsfonts,bm}

\def\eqref#1{equation~\ref{#1}}

\def\1{\bm{1}}

\DeclareMathAlphabet{\mathsfit}{\encodingdefault}{\sfdefault}{m}{sl}
\SetMathAlphabet{\mathsfit}{bold}{\encodingdefault}{\sfdefault}{bx}{n}

\usepackage{amsthm}
\usepackage{mathtools}
\usepackage{bbm}

\newtheorem{theorem}{Theorem} 
\newtheorem{lemma}[theorem]{Lemma}

\newtheorem{corollary}[theorem]{Corollary} 

\theoremstyle{definition}
\newtheorem{definition}{Definition}

\DeclarePairedDelimiterX{\infdivx}[2]{[}{]}{%
  #1\;\delimsize\|\;#2%
}
\DeclarePairedDelimiterX{\infdivmi}[2]{[}{]}{%
  #1\;;\;#2%
}
\newcommand{\MI}{I\infdivmi} 

\DeclarePairedDelimiter{\norm}{\lVert}{\rVert}
\DeclarePairedDelimiter{\abs}{\lvert}{\rvert}
\DeclarePairedDelimiterX{\innerProd}[2]{\langle}{\rangle}{%
    #1,#2%
}

\def\Oh{\mathcal{O}}

\newcommand{\Reals}{\mathbb{R}}

\newcommand{\Nats}{\mathbb{N}}

\newcommand{\defeq}{\stackrel{def}{=}}
\newcommand{\Prob}{\mathbb{P}}

\newcommand{\Unif}{\mathrm{Unif}}
\newcommand{\Exp}{\mathbb{E}}

\newcommand{\Bernoulli}{\mathrm{Bern}}
\newcommand{\lvlset}{\mathcal{H}}

\newcommand{\GeomDist}{\mathrm{Geom}}

\newcommand{\Ent}{\mathbb{H}}

\DeclareMathOperator{\enc}{\mathrm{enc}}
\DeclareMathOperator{\dec}{\mathrm{dec}}
\DeclareMathOperator{\XSpace}{\mathcal{X}}
\DeclareMathOperator{\YSpace}{\mathcal{Y}}
\DeclareMathOperator{\ZSpace}{\mathcal{Z}}

\makeatletter
\newcommand{\removelatexerror}{\let\@latex@error\@gobble}
\makeatother

\usepackage{cleveref}
    
\begin{document}

\title{Singular Relative Entropy Coding\\ with Bits-Back Rejection Sampling
\thanks{The authors acknowledge financial support from Imperial College London through an Imperial College Research Fellowship grant awarded to GF. SH acknowledges support from the Natural Sciences and Engineering Research Council of Canada.}
}

\author{\IEEEauthorblockN{Gergely Flamich}
\IEEEauthorblockA{Imperial College London \\
London, UK \\
g.flamich@imperial.ac.uk}
\and
\IEEEauthorblockN{Spencer Hill}
\IEEEauthorblockA{Queen's University \\
Kingston, Canada \\
spencer.hill@queensu.ca}
}

\maketitle
\begin{abstract}
A relative entropy code for a source $X \sim P_X$ is a stochastic code that encodes random samples from a prescribed $P_{Y \mid X}$ using as few bits as possible.
A generalisation of entropy coding, it is a standard result that the minimum number of bits required to achieve this is at least the mutual information $\MI{X}{Y}$.
However, a particularly fascinating feature of relative entropy coding compared to entropy coding is that, in general, this lower bound is only achievable to within an additional logarithmic factor.
As such, an important research direction is to identify channels where we can reduce this gap.
\par
Sriramu and Wagner \cite{sriramu2024optimal} achieved such success by exhibiting a relative entropy code for so-called singular channels with sub-logarithmic asymptotic redundancy.
However, their code is quite involved and, sadly, cannot be implemented in practice.
\par
In this paper, we construct the bits-back rejection sampler (BBRS), a relative entropy code that combines ideas from bits-back coding and (greedy) rejection sampling.
Our analysis of BBRS reveals that the algorithm achieves the same asymptotic
efficiency as Sriramu and Wagner's sampler, but with much simpler analysis and
better constants. Moreover, BBRS can be implemented using standard relative entropy coding methods. 
\end{abstract}


\section{Introduction}
\par
Given a pair of dependent random variables $X, Y \sim P_{X, Y}$ defining the channel $X \to Y$, the goal of channel simulation is to find a \textit{stochastic code} \cite{flamich2026data} that, given an input $X \sim P_X$, encodes a single sample $Y \sim P_{Y \mid X}$ in finitely many bits.
More formally, a stochastic code is a triplet $(Z, \enc, \dec)$, where
\begin{itemize}
\item $Z$ is a random variable independent of $X$ shared between the sender and receiver, called the \textit{common randomness}.
\item $\enc: \XSpace \times \ZSpace \to \{0, 1\}^*$ is the \textit{encoder} mapping an
  input and a realisation of $Z$ to a variable-length binary code.
\item $\dec: \{0, 1\}^* \times \ZSpace \to \YSpace$ is the \textit{decoder}
  mapping a variable-length binary string and a realisation of $Z$ to the output space $\YSpace$.
\end{itemize}
The basic identity that an (exact) stochastic code
must satisfy is that for all inputs $x \in \XSpace$,
\begin{align}
\dec(\enc(x, Z), Z) \sim P_{Y \mid X = x}
\end{align}
\par
It is a standard result that the average codelength of any stochastic code
must obey the bound \cite{harsha2010communication, li2024channel}
\begin{align}
\label{eq:stochastic_code_length_lower_bound}
  \MI{X}{Y} \leq \Ent[Y \mid Z] \leq \Exp[\abs{\enc(X, Z)}]
\end{align}
Beginning with the work of Harsha et al. \cite{harsha2010communication}, there has been a natural line of work to construct stochastic codes with cost as close as possible to the mutual information lower bound in \Cref{eq:stochastic_code_length_lower_bound}.
A significant milestone along this path is the work of Li and El Gamal \cite{li2018strong}, who describe a general scheme for constructing a stochastic code with expected codelength
\begin{align}
\label{eq:pfr_expected_codelength_bound}
\Exp[\abs{\enc(X, Z)}] \leq \MI{X}{Y} + \log(\MI{X}{Y} + 1) + 4 \end{align}
Sadly, we cannot do better than \Cref{eq:pfr_expected_codelength_bound} in general: there exist channels for which any stochastic code requires at least ${\MI{X}{Y} + \log(\MI{X}{Y} + 1) - \Oh(1)}$ bits \cite{braverman2014public,li2018strong}.
As such, following \cite{flamich2024data,flamich2026data}, we call stochastic codes that achieve a rate of $\MI{X}{Y} + \Oh(\log (\MI{X}{Y}))$ \textit{relative entropy codes}.
\par 
The necessity of the logarithmic gap between the lower and upper bounds prompts a natural follow-up question: are there any broad classes of channels for which we can either prove a tighter lower bound or construct a relative entropy code with better rate?
This question is precisely the focus of the works of Sriramu and Wagner \cite{sriramu2024optimal} and Flamich, Sriramu and Wagner \cite{flamich2025redundancy}, and indeed of the present paper.
These papers consider the ``asymptotic setting,'' meaning that for a pair of dependent random variables $X, Y \sim P_{X, Y}$ and product channel ${X^n \to Y^n}$, they study the best achievable rate
\begin{align}
R_n \defeq \frac{1}{n}\inf_{(Z, \enc, \dec)} \Exp[\abs{\enc(X^n, Z)}]
\end{align}
By the construction of the independent and identically distributed (iid) channel, we have that $\MI{X^n}{Y^n} = n \MI{X}{Y}$.
Combining this fact with \Cref{eq:stochastic_code_length_lower_bound,eq:pfr_expected_codelength_bound}, we have that $R_n \to \MI{X}{Y}$ as $n \to \infty$ and that this is optimal.
Thus, the papers refine the analysis and study the asymptotic \textit{logarithmic redundancy} 
\begin{align}
\label{eq:asymptotic_logarithmic_redundancy}
R^{\log} \defeq \lim_{n \to \infty} \frac{n (R_n - \MI{X}{Y})}{\log n}
\end{align}
Consulting \Cref{eq:stochastic_code_length_lower_bound,eq:pfr_expected_codelength_bound} once again, we immediately see that ${0 \leq R^{\log} \leq 1}$.
The main contribution of Sriramu and Wagner's work \cite{sriramu2024optimal} is identifying a simple condition called \textit{singularity} (see \Cref{sec:background} for the definition) on the channel ${X \to Y}$ which completely determines $R^{\log}$.
To be precise, they construct a family of relative entropy codes $\{(Z_n^{SW}, \enc_n^{SW}, \dec_n^{SW})\}$ based on rejection sampling whose asymptotic logarithmic redundancy for any channel $X^n \to Y^n$ satisfies
\begin{align}
R^{\log}_{SW} \defeq \lim_{n \to \infty} \frac{\Exp[\abs{\enc_n^{SW}(X^n, Z_n)}] - n \MI{X}{Y}}{\log n} \leq \frac{1}{2}
\end{align}
Furthermore, when the channel $X \to Y$ is singular, they show that $0 \leq R^{\log} \leq R^{\log}_{SW} = 0$.
This is complemented by further analysis in \cite{sriramu2024optimal,flamich2025redundancy}, which shows that when $X \to Y$ is non-singular, we have the lower bound $\frac{1}{2} \leq R^{\log}$.
Combining these results, we have the following clean characterisation of asymptotic logarithmic redundancy \cite{sriramu2024optimal,flamich2025redundancy}:
\begin{align}
\label{eq:asympt_log_redundancy}
R^{\log} =
\begin{cases}
0 &\text{when } X \to Y \text{ singular} \\
\frac{1}{2} &\text{when } X \to Y \text{ non-singular}.
\end{cases}
\end{align}
\par
However, the Sriramu-Wagner code for the achievability of \Cref{eq:asympt_log_redundancy} has three clear shortcomings: 
\begin{enumerate}
\item 
It relies on access to quantities that, despite being theoretically computable, are virtually impossible to compute, even for toy problems.
This intractability immediately precludes any practical utility of the code.
\item 
Adding to the impracticality, the construction exhibits a great deal of ``one-shot'' inefficiency: its rate has several sublogarithmic terms and large constants that only vanish in the limit $n \to \infty$.
\item
The reason \textit{why} singular channels achieve zero asymptotic logarithmic redundancy is not very clear. Essentially, singularity causes certain terms to cancel in the rate analysis of the code, but there is no clear interpretation of why one ought to expect this to occur.
\end{enumerate}
\par
This paper addresses these shortcomings by constructing a stochastic code, which we call the \textit{bits-back rejection sampler} (BBRS), and by providing a new proof of the achievability part of \Cref{eq:asympt_log_redundancy} in the singular case.
As we shall see, singularity plays a clear, central role in the construction of BBRS. Furthermore, although BBRS relies on the somewhat more advanced greedy rejection sampler, its one-shot rate is much better than the rate of the Sriramu-Wagner code.
Finally, the quantities required to implement BBRS are easier to compute in practice, facilitating implementation.
\subsection{Notation}
\par
In this paper, $\log$ denotes the binary logarithm.
We denote the set of finite-length binary strings as $\{0, 1\}^*$, and for a string $s \in \{0, 1\}^*$ we denote its length as $\abs{s}$.
Capital Roman letters denote random variables and probability measures.
Unless stated otherwise, we assume that all random variables are over some arbitrary Polish space.
For two random variables $X, Y$ and a probability measure $P$, we use a slight abuse of notation: $X \sim Y$ shall mean that $X$ and $Y$ are equal in distribution, while $X \sim P$ denotes that $X$ has probability distribution $P$.
For two probability measures $Q, P$ where $Q \ll P$, $dQ/dP$ denotes the Radon-Nikodym derivative of $Q$ with respect to $P$.
Bernoulli distributions with bias/success probability $p \in [0, 1]$ are denoted by $\Bernoulli(p)$, while geometric distributions with mean $M$ are denoted by $\GeomDist(M)$.
The Shannon entropy of a discrete random variable is denoted as $\Ent[X]$ and the mutual information between two variables $X, Y$ is denoted by $\MI{X}{Y}$.
\section{Background}
\label{sec:background}
\par
The five key concepts we require for our construction are \textit{singular channels}, \textit{rejection sampling}, \textit{greedy rejection sampling}, \textit{invertible sampling} and  \textit{bits-back coding}; we shall now cover them in this order.
\subsection{Singular Channels}
The basic motivation for singular channels is the observation that there exist certain channels for which all inputs that lead to a given output do so with the same probability. The generalisation for arbitrary $P_{Y \mid X}$ is the condition that only the support of $Y \mid X$ is affected by $X$, but not the precise likelihood:
\begin{definition}[Singular channel \cite{altuug2014refinement,sriramu2024optimal}]
\label{def:singular_channel}
Let $X, Y \sim P_{X, Y}$ be a pair of dependent random variables.
Then, we say that the channel $X \to Y$ is singular if there exists a $P_Y$-measurable function $g$ such that
\begin{align}
\label{eq:singularity_condition}   \frac{dP_{Y \mid X}}{dP_Y}(y \mid x) = g(y)\quad  P_{X, Y}\text{-almost surely.}
\end{align}
\end{definition}
The two canonical examples of singular channels are the additive uniform channel ${Y = X + U}$, where ${U \sim \Unif(-1, 1)}$, and the binary erasure channel. Importantly, for the rest of this paper, we shall assume that the function $g$ is known and shared by the communicating parties; this will be crucial for our construction of BBRS.
\subsection{Rejection Sampling}
\label{sec:rejection_sampling}
Here we give a brief description of rejection sampling and how we may use it to construct a stochastic code; see \cite{flamich2026data} for a tutorial and \cite{li2024channel} for a comprehensive survey of these concepts.
\par
Originating in computational statistics, rejection sampling aims to simulate a sample from some prescribed target distribution $Q$ given access to iid samples from some proposal distribution $P$.
Concretely, let $r = dQ/dP$ and assume that $r$ is bounded, that is, there exists $M \geq 1$ such that $\norm{r}_\infty \leq M$.
Then, at each step $k = 1, 2, \dotsc$, rejection sampling considers a sample $Z_k \sim P$, computes the \textit{acceptance probability} $\alpha(Z_k) = r(Z_k) / M$, and flips a coin $B_k \sim \Bernoulli(\alpha(Z_k))$.
If $B_k = 1$, the algorithm \textit{accepts} $Z_k$ and terminates, while if $B_k = 0$, it \textit{rejects} $Z_k$ and moves on to step $k + 1$.
Letting $K$ be the step in which the algorithm terminates, it is a standard result that $Z_K \sim Q$.
\par
Given a pair of dependent random variables $X, Y \sim P_{X, Y}$, we can use rejection sampling to construct a stochastic code for the channel $X \to Y$: given an input $x \sim P_X$, we use rejection sampling to simulate a sample from the target $P_{Y \mid X}$ using $P_Y$ as the proposal distribution.
First, we set the common randomness as an iid sequence of $P_Y$-distributed proposal samples:
$Z \gets \{Z_i\}_{i = 1}^\infty$.
Then, letting $r_x(y) = \frac{dP_{Y \mid X}}{dP_Y}(y \mid x)$ with upper bound $M_x$, given some input $x$ and the common randomness $Z$, we run rejection sampling to compute
\begin{align}
\label{eq:rs_accepted_index}
K(x, Z)\! =\! \min\{k \in \Nats \mid B_k = 1,\, B_k\! \sim\! \Bernoulli(\alpha(x, Z_k))\}
\end{align}
with the acceptance probability $\alpha(x, z) = r_x(z) / M_x$.
Given that each proposed sample $Z_k$ and acceptance decision $B_k$ is iid, we find that $K(x, Z) \sim \GeomDist(M_x)$, that is, the acceptance step is geometrically distributed with mean $M_x$.
Thus, assuming that both the sender and the receiver know $M_x$ (an unrealistic assumption in general, but as we shall see, for BBRS it will not be an issue), we can construct an entropy code $C$ (for example, a Huffman code) for $K(x, Z)$ whose average codelength is at most 
\begin{align}
\label{eq:rs_rate}
\Exp[\abs{C(K(x, Z))}] \leq \Ent[K(x, Z) \mid M_x] \!+\! 1 \leq \log M_x \!+\! 2
\end{align}
Hence, our final rejection code for $X \to Y$ is
\begin{equation}
\label{eq:rs_code}
\begin{aligned}
Z^{RS} &= \{Z_i\}, \quad Z_i \sim P_Y\\
\enc^{RS}(x, Z^{RS}) &= C(K(x, Z^{RS})) \\
\dec^{RS}(s, Z^{RS}) &= Z_{C^{-1}(s)} \quad Z_i \in Z^{RS}
\end{aligned}
\end{equation}
The issue with the stochastic code in \Cref{eq:rs_code} is twofold: first, the required upper bound $M_x$ is usually unknown to the decoder.
On the other hand, if we use a universal upper bound $M \geq \sup_{x} M_x$, the rate will be much larger than the optimum.
This issue is precisely why Harsha et al. developed greedy rejection sampling, which we shall discuss next.
\subsection{Greedy Rejection Sampling (GRS)}
\label{sec:grs}
\begin{algorithm}[t]
\SetAlgoLined
\DontPrintSemicolon
\SetKwInOut{Input}{Input}
\SetKwInOut{Output}{Output}
\Input{Target $Q$, Shared iid samples $X_1, X_2\hdots \sim P$.}
$L_1, S_1 \gets (0, 1)$ \;
\For{$k = 1$ \KwTo $\infty$}{
$\alpha_k \gets \min\left\{1, \max\left\{0, \left(\frac{dQ}{dP}(X_k) - L_k \right) \Big / S_k\right\} \right\}$ \;
$B_k \sim \Bernoulli(\alpha_k)$
 \If{$B_k = 1$}{
    \KwRet{$X_k, k$}
 }
 $L_{k + 1} \gets L_k + S_k$ \;
 $\lvlset_{k + 1} \gets \left\{ y \in \Omega \mid L_{k + 1} \leq \frac{dQ}{dP}(y) \right\}$ \;
 $S_{k + 1} \gets Q(\lvlset_{k + 1}) - L_{k + 1} \cdot P(\lvlset_{k + 1})$
}
\caption{Greedy rejection sampling, as formulated in \cite{flamich2023adaptive}.
See \Cref{sec:background} for notation.}
\label{alg:grs}
\end{algorithm}
Intuitively, the reason for the inefficiency of the rate of the rejection code in \Cref{eq:rs_code} is that the acceptance decisions of rejection sampling are not only independent but also identically distributed.
The key idea behind the greedy rejection sampler \cite{harsha2010communication,flamich2023adaptive} is to choose the acceptance probabilities according to the following criterion: in each step $k$, assuming the sampler hasn't terminated in a previous step, choose the acceptance probability $\alpha_k(z)$ to maximise the probability that the sampler terminates in the current step.
Indeed, this design principle of greedily maximising the termination probability is the algorithm's namesake.
\par
We describe GRS in \Cref{alg:grs} for completeness; for a detailed discussion and analysis of the algorithm, see \cite{harsha2010communication,flamich2023adaptive}.
For this paper, there are two important facts the reader needs to know.
First, GRS is also a rejection procedure and, in essence, differs from regular rejection sampling only in that the acceptance probabilities are not identically distributed.
Therefore, we can similarly use GRS to construct a stochastic code for a channel $X \to Y$. 
Second, letting
\begin{align*}
K'(x, Z) = \min\{k \in \Nats \mid B_k = 1,\, B_k \sim \Bernoulli(\alpha_k(x, Z_k))\}
\end{align*}
as Theorem III.2 and in particular Equation 3 of \cite{flamich2023adaptive} shows, we can construct an entropy code $C'$ for $K'$ such that
\begin{align}
\label{eq:grs_rate}
\Exp[\abs{C'(K'(X, Z))}] \leq \MI{X}{Y} + \log (\MI{X}{Y} + 1) + 5
\end{align}
As such, GRS is a general way of constructing a relative entropy code.
\subsection{Invertible Sampling}
The next ingredient we will need for our construction is invertible sampling, which sits at the heart of bits-back coding \cite{hinton1993keeping,townsend2018practical,flamich2023adaptive}.
For this, we need to explicitly consider the bit stream $s$ that the sender intends to transmit to the receiver and, accordingly, extend the definition of our codes. 
Concretely, henceforth we shall assume that for any discrete source distribution $P$ we can construct a stream code $C_P: \{0, 1\}^* \times \XSpace \to \{0, 1\}^*$ such that for any message stream $s \in \{0, 1\}^*$ and source $x \in \XSpace$, we have
\begin{align}
C^{-1}_P(C_P(s, x)) = (s, x)
\end{align}
and 
\begin{align}
\abs{C_P(s, x)} - \abs{s} \approx -\log P(x)
\end{align}
Examples of such stream codes are arithmetic coding \cite{witten1987arithmetic} and asymmetric numeral systems (ANS) \cite{duda2015use,bamler2022understanding}.
\par
The main idea behind invertible sampling is the answer to the following puzzle.  
Consider a very long string $S$ whose bits are each iid $\Bernoulli(1/2)$ and a discrete distribution $P$.
If we decode $(X, S') \gets C_P^{-1}(S)$, what is the distribution of $X$?
To find the intuitive answer, consider the reverse scenario first: encoding some $X' \sim P$ into the stream of uniform bits $S'$.
As we assumed that $C_P$ is a stream code and we are encoding a source sample $X'$ using its own distribution into $S'$, $C_P(S', X')$ ought to be a string of uniformly random bits.
Thus, if we now invert the procedure and start with $S$, and put $(X, S') = C_P^{-1}(S)$, then we ought to have $X \sim P$.
These arguments can be made precise, for example, see \cite[Chapter 13.3]{cover2012elements}, \cite{bamler2022understanding} and \cite{knuth1976complexity}.
\par
The reason why we call this notion \textit{invertible sampling} is because starting from a pseudo-random number generation (PRNG) perspective, we may consider the initial stream $S$ as the random seed or random state of a PRNG algorithm and the decoding procedure $(X, S') \gets C_P^{-1}(S)$ as simulating a random sample and advancing the random seed/state of the algorithm.
However, we can do more than usual: we can invert the sampler and recover the original seed $S$ by re-encoding $X$ into the stream $S'$ using $C_P$!
\subsection{Bits-Back Coding}
\label{sec:bits-back}
The purpose of bits-back coding is to construct a stream code for some source $X \sim P_X$ where $P_X$ is difficult to compute, but there exists a channel $P_{Y\mid X}$ and function $g$ such that it is easy to perturb $X$ using $X \to Y$ and recover $X$ by computing $g(Y)$.
Indeed, a high-level description of bits-back coding is \cite{flamich2024data}
\begin{align*}
\text{invertible sampling } +  \text{ error correction } = \text{bits-back coding}.
\end{align*}
More concretely, as a toy example, consider $X$ to be uniformly distributed over the codewords of the $(7, 4, 3)$-Hamming code and let $Y = X + \epsilon$, where $\epsilon$ is uniformly distributed over the Hamming ball of radius $1$ over $7$-bit strings.
Observe that we constructed this channel such that we can always recover $X$ from $Y$ by the very definition of Hamming codes; indeed, $g$ is the Hamming decoder.
Furthermore, we see that the marginal $P_Y$ is the uniform distribution over $7$-bit strings.
\par
Consider now constructing an entropy code for $X$.
Of course, in this toy setting, we could do this directly, but note that we can also take advantage of the error-correction procedure combined with invertible sampling as follows:
\begin{enumerate}
\item Assume the encoder has stream $S$ and receives a new symbol $X \sim P_X$ to encode.
\item They \textbf{decode} $(S', Y) \gets C_{P_{Y \mid X}}^{-1}(S)$.
Then, $Y \sim P_{Y \mid X}$.
\item They \textbf{encode} $S'' \gets C_{P_Y}(S', Y)$ and send $S''$.
\end{enumerate}
Somewhat magically, this procedure is completely invertible:
\begin{enumerate}
\item Upon receiving $S''$, the decoder uses the shared coding distribution $P_Y$ to recover ${(S', Y) \gets C_{P_Y}^{-1}(S'')}$
\item They ``error-correct'' $Y$ by computing $X \gets g(Y)$.
\item Finally, they recover the original stream by re-encoding $Y$ into $S'$: they compute $S \gets C_{P_{Y \mid X}}(S', Y)$, which they can do, since they now have access to $P_{Y \mid X}$.
\end{enumerate}
Note that at the end of the decoding procedure, the decoder has obtained the source $X$ as well as the original stream state $S$.
Furthermore, the encoding procedure is also efficient:
decoding $(S', Y)$ from the stream $S$ decreases the length of the stream by $\approx -\log P_{Y \mid X}(Y \mid X)$ bits, while encoding $Y$ into $S'$ using $P_Y$ increases the stream length by $\approx -\log P_{Y}(Y)$ bits.
Thus, in total, the stream length increases by roughly
\begin{align*}
-\log \frac{P_Y(Y)}{P_{Y \mid X}(Y \mid X)} = -\log \frac{P_X(X)}{P_{X \mid Y}(X \mid Y)} = -\log P_X(X)
\end{align*}
bits, where the first equality follows from Bayes' rule and the second follows from the fact that $X$ is a deterministic function of $Y$ and hence the denominator is simply equal to $1$.
\par
The most widespread application of bits-back coding is in the construction of the ANS stream code \cite{bamler2022understanding}, but it is also used in shuffle coding, a state-of-the-art entropy code for unordered/exchangeable data \cite{kunze2024entropy}, and bits-back quantisation, a relative entropy code for truncated distributions over the reals \cite{flamich2023adaptive}.
Finally, we note that bits-back coding can also be applied when $X$ cannot be fully recovered from $Y$: assuming we can only recover $\hat{X} = g(Y)$, we can then encode $X \mid \hat{X}$. 
However, we note that in this case, the procedure is no longer rate-optimal.
Indeed, this was the setting in which bits-back coding was first introduced to practical machine learning-based lossless compression by Townsend et al. \cite{townsend2018practical}.
\section{Bits-Back Rejection Sampling}
\par
We are now finally ready to present our stochastic code for singular channels, bits-back rejection sampling (BBRS).
\par
The high-level idea is very similar to that of the Sriramu-Wagner construction: instead of encoding a sample $Y \sim P_{Y \mid X}$ directly, we produce a two-part code by first encoding the quantised log-density ratio $\Gamma = Q_\Delta\left(\log \frac{dP_{Y \mid X}}{P_Y}(Y \mid X)\right)$ to $\Delta$ precision, and then encode a sample $Y \mid X, \Gamma$. As elucidated in~\cite{sriramu2024optimal}, the second sample can be communicated optimally using standard rejection sampling, as the density ratio $\Gamma$, and therefore the acceptance probability, is fixed and known to both encoder and decoder.
\par
However, our construction is based on the observation that when $X \to Y$ is singular, we can use the measurable function $g$ from \Cref{def:singular_channel} as an ``error-correction mechanism'' to recover $\Gamma$ from $Y$. Therefore, we can apply bits-back coding to recover almost all the bits we used to encode $\Gamma$!
\par
Concretely, let $P_{X, Y}$ be given with $X \to Y$ singular.
For $c \in \Reals$ and $\Delta > 0$, let $Q_\Delta(c) = \Delta \cdot \lfloor c/\Delta\rfloor$ be a $\Delta$-fine quantiser for the reals.
Let $\Gamma = Q_{\Delta}\left(\log \frac{dP_{Y \mid X}}{dP_Y}(Y \mid X)\right)$.
Then, the encoding procedure is:
\begin{enumerate}
\item Given $X \sim P_X$, the encoder simulates $\Gamma \sim P_{\Gamma \mid X}$.
\item Next, they use a GRS-based relative entropy code with $P_Y$-distributed shared randomness $\{Z_k\}_{k = 1}^\infty$ to encode a sample $Y' \sim P_{Y \mid \Gamma}$.
Concretely, they encode GRS's returned index $K$ using an appropriate entropy code.
\item The encoder knows that given $Y'$, the singularity of $X \to Y$ allows the decoder to recover $\Gamma$ by computing $\Gamma = Q_\Delta(\log g(Y'))$.
Thus, they use a simple rejection code to encode a sample from $Y \mid X, \Gamma$ using $P_{Y \mid \Gamma}$-distributed common randomness $\{\Upsilon_n\}_{n = 1}^\infty$.
To implement the rejection code, they compute two quantities:
\begin{align}
M(\Gamma) &= 2^{\Gamma + \Delta}\\
M(\Gamma, X) &= M(\Gamma) \Big/ \frac{dP_{X \mid \Gamma}}{P_X}(X \mid \Gamma)
\end{align}
It is not hard to show that $M(\Gamma, X)$ is an upper bound to $\frac{dP_{Y \mid X, \Gamma}}{dP_{Y \mid \Gamma}}$ and can therefore be used to run the rejection sampler \cite{sriramu2024optimal}; we provide a proof of this fact in \Cref{app:bbrs_technical_details} for completeness.
Finally, they encode the rejection sampler's returned index $N$ using an entropy code for the geometric distribution with mean $M(\Gamma)$. 
\end{enumerate}
The encoder also knows that once the decoder recovers $\Gamma$, they can ``replay'' the GRS procedure for sampling $Y \mid \Gamma$.
This observation then motivates a bits-back procedure: given the encoder's message stream $s$, the encoder \textbf{decodes the GRS acceptance decisions} from $s$, thus shortening the message.
They can do this because they know the decoder can re-encode the acceptance decisions back into the stream once it has recovered $\Gamma$.
Although the above is a full description of BBRS, for clarity, we also provide detailed pseudocode for the encoder and decoder in \Cref{alg:bbgrs_encoder,alg:bbgrs_decoder}, respectively.
We use the ANS stream code \cite{duda2015use,bamler2022understanding} as our invertible sampler in the pseudocode, which is a last-in-first-out code: the last symbol we encode into the stream is the first symbol we decode from it.
Note that \Cref{alg:bbgrs_encoder,alg:bbgrs_decoder} account for this by encoding the two-part code in the appropriate order.
\subsection{Rate Analysis of Bits-Back Rejection Sampling}
We now turn to the analysis of the BBRS rate and obtain 
\begin{theorem}[One-shot rate of BBRS]
\label{thm:singular_bbrs_one_shot_rate}
Let $X, Y \sim P_{X, Y}$ be a pair of dependent random variables with $X \to Y$ singular.
Let $(Z, \enc, \dec)$ be the stochastic code defined by \Cref{alg:bbgrs_encoder,alg:bbgrs_decoder}.
Then, we have
\begin{align}
\Exp[\abs{\enc(X, Z)}] \leq \MI{X}{Y} + \log (\Ent[\Gamma] \!+\! 1) + 2\Delta + 5
\end{align}
\end{theorem}
\begin{proof}
The rate of the algorithm consists of three parts: 
\begin{enumerate}
\item The number of bits required to encode the GRS index $K$, which encodes the sample $Y' \sim P_{Y \mid \Gamma}$.
\item The index $N$ which encodes $Y \sim P_{Y \mid X, \Gamma}$.
\item The number of bits by which the message is reduced from decoding the GRS acceptance decisions.
\end{enumerate}
Thus, we now analyse these quantities in this order.
\par
\textbf{Rate required for $\bm{K}$.}
From \Cref{sec:grs}, and \Cref{eq:grs_rate} in particular, assuming the sender and receiver both know $\MI{Y}{\Gamma}$, the number of bits required to encode $K$ is at most
\begin{align*}
\MI{Y\!}{\!\Gamma} + \log (\MI{Y\!}{\!\Gamma} + 1) + 4
= \Ent[\Gamma] + \log (\Ent[\Gamma] + 1) + 4
\end{align*}
where the equality holds as $\Gamma$ is a function of $Y$.
\par
\textbf{Rate required for $\bm{N}$.}
As we explained in \Cref{sec:rejection_sampling}, since we use standard rejection sampling with upper bound $M(X, \Gamma)$, $N$ will be geometrically distributed with mean $M(X, \Gamma)$.
Recall that the sender encodes $N$ using a geometric random variable with mean $M(\Gamma)$.
Then, noting that $\Exp[M(\Gamma, X) \mid \Gamma] = M(\Gamma)$, the average cost of encoding $N$ will be equal to the expected cross entropy
\begin{align*}
\Exp_X[&\Ent[\GeomDist(M(X, \Gamma))\,\Vert\, \GeomDist(M(\Gamma))]]\\
&= \Exp_{X, \Gamma}[\log M(\Gamma) - (M(X, \Gamma) - 1)\log(1 - 1/M(\Gamma))] \\
&= \Exp_{\Gamma}\left[\log M(\Gamma) - \left(M(\Gamma) - 1\right)\log(1 - 1/M(\Gamma))\right] \\
&= \Ent[\GeomDist(M(\Gamma))] \\
&\leq \Exp[\log M(\Gamma)] + 1 \\
&= \Exp[\Gamma + \Delta] + 1 \\
&\leq \MI{X}{Y} + 2\Delta + 1
\end{align*}
\par
\textbf{Number of bits decoded to sample $\bm{Y' \sim P_{Y \mid \Gamma}}$.}
Letting $A_k$ denote the acceptance probabilities of GRC, for a fixed $K = k$ the number of decoded bits $H$ is
\begin{align*}
H 
= -\log A_k -\sum_{i = 1}^{k - 1} \log (1 - A_i) 
= -\log\left(\!A_k \prod_{i = 1}^{k - 1} (1 - A_i)\!\right)
\end{align*}
However, note that $A_k \prod_{i = 1}^{k - 1} (1 - A_i) = \Prob[K = k \mid \{Z_k\}_{k = 1}^\infty]$.
Hence, $\Exp[H] = \Ent[K \mid \{Z_k\}_{k = 1}^\infty] \geq \MI{Y}{\Gamma} = \Ent[\Gamma]$, where the inequality follows from \Cref{eq:stochastic_code_length_lower_bound} and the last equality follows since $\Gamma$ is a deterministic function of $Y$.
\par
\textbf{Putting it all together.}
Summing the upper bounds on the rates for $K$ and $N$ and subtracting the lower bound on the number of decoded bits yields the desired result.
\end{proof}
To recover Sriramu and Wagner's upper bound, we use their Lemma 5 from \cite{sriramu2024optimal}, which essentially follows directly from the central limit theorem applied to  $\log\frac{dP_{Y^n \mid X^n}}{dP_Y^n}(Y^n \mid X^n)$:
\begin{lemma}
\label{lemma:gamma_asymp_ent}
Let $\Gamma_n = Q_\Delta\left(\log\frac{dP_{Y^n \mid X^n}}{dP_Y^n}(Y^n \mid X^n)\right)$.
Then, there is a constant $C > 0$ that depends on $P_{X, Y}$ and $\Delta$, but not $n$, such that
\begin{align}
\Ent[\Gamma_n] \leq \frac{1}{2}\log n + C
\end{align}
\end{lemma}
Thus, combining \Cref{thm:singular_bbrs_one_shot_rate} with Lemma~\ref{lemma:gamma_asymp_ent}, we immediately obtain the desired corollary.
\begin{corollary}
Let $X, Y \sim P_{X, Y}$ be a pair of dependent random variables with $X \to Y$ singular.
Consider the family of stochastic codes $\{(Z_n, \enc_n, \dec_n)\}_{n = 1}^\infty$ with $(Z_n, \enc_n, \dec_n)$ the stochastic code formed by applying \Cref{alg:bbgrs_encoder,alg:bbgrs_decoder} sequentially to encode $n$ copies of $X \to Y$.
Then, the logarithmic redundancy of this family is 
\begin{align}
\lim_{n \to \infty}\frac{\Exp[\abs{\enc_n(X^n, Z_n)}] - n\MI{X}{Y}}{\log n} = 0
\end{align}
which also implies that the optimal logarithmic redundancy is $R^{\log} = 0$ as well.
\end{corollary}
Finally, we note that BBRS can be extended to non-singular channels as well by encoding $\Gamma \mid Y$; this is akin to the ``encoding the residuals'' idea we discussed in \Cref{sec:bits-back}.
Sadly, this scheme is no longer efficient, as it does not even achieve the $ 1/2$-logarithmic redundancy asymptotically, and thus we do not discuss it further.
\par
However, note that our main motivation for developing BBRS was to give a simpler code than the Sriramu-Wagner construction, but the non-singular case already has simple algorithms that achieve the optimum. 
As Sriramu and Wagner point out in their paper, in an observation attributed to Lucas Theis \cite{sriramu2024optimal}, the $1/2$-logarithmic redundancy can also be achieved by a small modification to existing relative entropy codes, such as GRS or the Poisson functional representation (PFR).
However, as they do not provide an explicit construction, we provide one for PFR in \Cref{app:pfr_asymp_coding} for completeness.
\section{Conclusion}
\par
We constructed a stochastic code for singular channels using a sampling algorithm we call the bits-back rejection sampler.
We analysed its rate and showed that in the asymptotic case, its logarithmic redundancy is zero, thereby recovering Sriramu and Wagner's result for singular channels.
However, our sampler uses singularity more explicitly than their code by embedding it in a bits-back procedure. 
\par
One unsatisfactory aspect of our sampler, which it shares with Sriramu and Wagner's construction, is that it relies on some arbitrary discretisation of the log-density ratio. As such, future work might aim either to eliminate this quantisation or to analyse its precise role in the sampler's behaviour.
\bibliographystyle{ieeetr}
\bibliography{references}

\clearpage
\appendices
\crefalias{section}{appendix}
\section{Minor Technical Details for BBRS}
\label{app:bbrs_technical_details}
Here, we show that $M(X, \Gamma)$ is a valid upper bound on the density ratio of rejection sampling in the second stage. 
This verifies the correctness of the scheme.
Our proof follows closely that of~\cite[Lemma 3]{sriramu2024optimal}. 
We rearrange the Radon-Nikodym derivative: 
\begin{align}
&\frac{dP_{Y \mid X, \Gamma}}{dP_{Y \mid \Gamma}} (y \mid x, \gamma)  \nonumber \\
&\quad=  \frac{dP_{\Gamma \mid X, Y}}{dP_{\Gamma \mid Y}} (\gamma \mid x, y) \, \frac{dP_{Y \mid X}}{dP_{Y}}(y \mid x) \, \frac{dP_{\Gamma}}{dP_{\Gamma \mid X}}(\gamma \mid x) \label{eq:chainrule} \\
&\quad=  \frac{P_{\Gamma \mid X, Y}(\gamma \mid x, y)}{P_{\Gamma \mid Y}(\gamma \mid y)} \, \frac{dP_{Y \mid X}}{dP_{Y}}(y \mid x) \, \frac{P_{\Gamma}(\gamma)}{P_{\Gamma \mid X}(\gamma \mid x)} \label{eq:discrete} \\
&\quad= \frac{P_{\Gamma}(\gamma)}{P_{\Gamma \mid X}(\gamma \mid x)} \frac{dP_{Y \mid X}}{dP_{Y}}(y \mid x) \label{eq:singular} 
\end{align}
In the above,~\Cref{eq:chainrule} follows by the chain rule for the Radon-Nikodym derivative;~\Cref{eq:discrete} follows because $\Gamma$ is discrete and therefore $P_{\Gamma}$, $P_{\Gamma \mid X}, P_{\Gamma \mid Y}$ and $P_{\Gamma \mid X, Y}$ are absolutely continous with respect to the counting measure; and~\Cref{eq:singular} holds because $X \to Y$ is singular and hence $P_{\Gamma \mid X, Y}$ and $P_{\Gamma \mid Y}$ are delta functions. 
Then:
\begin{align}
\sup_{ y \, : \, \Gamma = \gamma}\frac{dP_{Y \mid X, \Gamma}}{dP_{Y \mid \Gamma}} &(y \mid x, \gamma) \\
&\leq M(\Gamma) \Big/ \frac{dP_{X \mid \Gamma}}{P_X}(X \mid \Gamma) = M(\Gamma, X). \nonumber 
\end{align}
\section{Improved Poisson Functional Representation Coding scheme}
\label{app:pfr_asymp_coding}
Here, we give an explicit construction that the Poisson functional representation (PFR) can be moderately modified so that its code length attains the desired $1/2$-logarithmic redundancy. 
We do not describe PFR here; see the original paper \cite{li2018strong} for a description or \cite{li2024channel,flamich2024data,flamich2026data} for a more detailed account.
The reader only needs to know that, akin to rejection sampling and GRS, PFR also uses $P_Y$-distributed common randomness $\{Z_k\}_{k = 1}^\infty$ as proposal samples for the target distribution $P_{Y \mid X}$ and encodes the index $K$ of the selected sample.
Then, letting 
\begin{align}
R_x &= \frac{dP_{Y \mid X}}{dP_Y}(Z_K \mid x) \\
\Gamma_x &= Q_\Delta\left(\log R_x \right)\\
M(\Gamma_x) &= 2^{\Gamma_x + \Delta} + 1\\
M'(R_x) &= \Exp_{Z \sim P_Y}\left[\max\left\{\frac{dP_{Y \mid X}}{dP_Y}(Z \mid x), R_x\right\}\right]
\end{align}
we have the following result \cite{li2021unified,li2024channel,flamich2024data}:
\begin{align}
K \mid R_x \sim \GeomDist\left(M'(R_x)\right)    
\end{align}
Thus, consider now encoding $K$ using a geometric distribution with mean $M(\Gamma_x)$.
Then, the rate will be equal to the expected cross-entropy
\begin{align*}
\Exp_X[&\Ent[\GeomDist(M'(R_X))\,\Vert\, \GeomDist(M(\Gamma_X))]]\\
&= \Exp[\log M(\Gamma_X) - (M'(R_X) - 1)\log(1 - 1/M(\Gamma_X))]
\end{align*}
Treating these terms separately, for the first term, we see that 
\begin{align*}
&\Exp[\log M(\Gamma_X)] \\
&= \Exp[\log (2^{\Gamma_X + \Delta} + 1)]\\
&\leq \Exp[\log (R_X \cdot2^{\Delta} + 1)] \\
&= \Exp[\log R_X] + \Exp[\log (2^{\Delta} + 1/R_X)] \\
&\leq \Exp[\log R_X] + \Exp[\log (2^{\Delta}e^{1/R_X})] \\
&=\MI{X}{Y} + \Delta + \log e \cdot\Exp_{X, Y \sim P_{X, Y}}\left[\frac{dP_{Y}}{dP_{Y \mid X}}(Y \mid X)\right] \\
&= \MI{X}{Y} + \Delta + \log e
\end{align*}
For the second term, we first note that 
\begin{equation}
\label{eq:pfr_code_second_term_first_ineq}
\begin{aligned}
M'(R_x) 
&= \Exp_{Z \sim P_Y}\left[\max\left\{\frac{dP_{Y \mid X}}{dP_Y}(Z \mid x), R_x\right\}\right] \\
&\leq R_x + 1    
\end{aligned}
\end{equation}
and also that $R_x + 1 \leq M(\Gamma_X)$, which implies 
\begin{equation}
\label{eq:pfr_code_second_term_second_ineq}
\begin{aligned}
-\log\left(1-\frac{1}{M(\Gamma_x)}\right) &\leq  -\log\left(1-\frac{1}{R_x + 1}\right)
\end{aligned}    
\end{equation}
Applying \Cref{eq:pfr_code_second_term_first_ineq,eq:pfr_code_second_term_second_ineq} to the second term, we see that it is bounded above as
\begin{align*}
-\Exp[ &(M'(R_X) - 1)\log(1 - 1/M(\Gamma_X))] \\
&\leq 
-\Exp\left[ ((R_X + 1) - 1) \cdot \log\left(1-\frac{1}{R_X + 1}\right)\right]\\
&\leq \log e
\end{align*}
where the second inequality follows from the standard pointwise inequality: for $x \in (0, 1)$, we have $-x \ln x \leq 1$.
Finally, this shows that the cross-entropy is bounded above by
\begin{align}
\Ent[\GeomDist(M'(R_X))\,\Vert\,& \GeomDist(M(\Gamma_X))] \nonumber\\
&\leq \MI{X}{Y} + \Delta + 2\log e
\end{align}
Thus, we can construct a two-part code by first encoding $\Gamma_X$ followed by $K \mid R_X$.
The total rate will thus be less than
\begin{align*}
\MI{X}{Y} + \Ent[\Gamma] + \Delta + 2\log e + 1
\end{align*}
From this, we see that in the asymptotic setting, the only quantity contributing to the logarithmic redundancy is $\Ent[\Gamma_n]$, which by Lemma 5 of Sriramu and Wagner grows as $\frac{1}{2}\log n + \Oh(1)$, as desired.

\clearpage
\begin{figure}[H]
\centering
\begin{minipage}[H]{.48\textwidth}
\removelatexerror
\begin{algorithm}[H]
\SetAlgoLined
\DontPrintSemicolon
\SetKwIF{If}{ElseIf}{Else}{if}{:}{else if}{else}{end}
\SetKwFunction{ANSdecode}{ANSdecode}
\SetKwFunction{ANSencode}{ANSencode}
\SetKwFunction{loopBreak}{break}
\SetKwInOut{Input}{Input}\SetKwInOut{Output}{Output}
\textbf{Input:}\;
Channel input $X = x$\;
Common randomness $\{Z_k\}$ with $Z_k \sim P_Y$ for each $k$\;
Collection of common randomness $\{ \{ \Upsilon_n \}_{\Gamma = \gamma} \}_\gamma$, where for each sequence $\{\Upsilon_n\}_{\Gamma = \gamma}$, we have $\Upsilon_n \sim P_{Y \mid \Gamma = \gamma}$ for each $n$\;
Message bit-stream $s$\;
\textbf{Output:}\;
Message bit-stream $s$ with two-part code $K, N$ embedded in it with $Z_K \sim P_{Y \mid \Gamma}$ and $\Upsilon_N \sim P_{Y \mid X, \Gamma}$\;
\;
\tcp{Stage 1: Sample the quantised density ratio}
$\hat{Y} \sim P_{Y \mid X=x}$\;
$\Gamma \gets Q_\Delta\left(\log \frac{dP_{Y\mid X}}{dP_Y}(\hat{Y} \mid x)\right)$\;
\;
\tcp{Stage 2: GRS to obtain an index $K$ such that $Z_K \sim P_{Y \mid \Gamma}$}
$k, b \gets 1, 0$\;
\;
\While{$b = 0$}{
\tcp{Compute GRS acceptance probability. Note that $S_k$ and $L_k$ only depend on $ \frac{dP_{Y \mid \Gamma}}{dP_Y}$}
$A_k \gets \min\left\{1, \max\left\{0, \left(\frac{dP_{Y \mid \Gamma}}{dP_Y}(Z_{k}) - L_k\right) \Big/ S_k \right\}\right\}$\;
\;
\tcp{Obtain acceptance decision and updated stream with invertible sampler, shortening the stream}
$s, B_{k} \gets \ANSdecode(s \mid \Bernoulli(A_k))$\;
$b \gets B_{k}$\;
$k \gets k + 1$
}
\;
\tcp{Stage 3: rejection sampling an index $N$ such that $\Upsilon_N \sim P_{Y \mid X, \Gamma}$}
$M(\Gamma) \gets 2^{\Gamma + \Delta}$\;
$M(\Gamma, X) \gets M(\Gamma) \Big/ \frac{dP_{X \mid \Gamma}}{dP_X}(X \mid \Gamma)$\;
Select $\{\Upsilon_n\}_{\Gamma}$ from the collection of common randomness, i.e. where $\Upsilon_n \sim P_{Y \mid \Gamma}$
\;
\For{$n = 1, 2, \hdots$}{
$U_n \sim \Unif(0, 1)$\;
\If{$U_n \leq \frac{dP_{Y \mid \Gamma, X}}{dP_{Y \mid \Gamma}}(\Upsilon_n \mid \Gamma, x) \Big/M(\Gamma, X)$}{
\tcp{Upon acceptance, encode $K, N$ into the stream}
$s \gets \ANSencode(s, N \mid \GeomDist(M(\Gamma)))$\;
$s \gets \ANSencode(s, K \mid \zeta)$ \quad\tcp{Use optimal zeta code}
\Return{$s$}
}
}
\caption{Bits-back rejection sampling encoder}
\label{alg:bbgrs_encoder}
\end{algorithm}
\end{minipage}%
\end{figure}%
\begin{figure}[H]
\centering
\begin{minipage}[H]{.48\textwidth}
\removelatexerror
\begin{algorithm}[H]
\SetAlgoLined
\DontPrintSemicolon
\SetKwIF{If}{ElseIf}{Else}{if}{:}{else if}{else}{end}
\SetKwFunction{ANSdecode}{ANSdecode}
\SetKwFunction{ANSencode}{ANSencode}
\SetKwFunction{loopBreak}{break}
\SetKwInOut{Input}{Input}\SetKwInOut{Output}{Output}
\textbf{Input:}\;
Message bit-stream $s$\;
Common randomness $\{Z_k\}$ with $Z_k \sim P_Y$ for each $k$\;
Collection of common randomness $\{ \{ \Upsilon_n \}_{\Gamma = \gamma} \}_\gamma$, where for each sequence $\{\Upsilon_n\}_{\Gamma = \gamma}$, we have $\Upsilon_n \sim P_{Y \mid \Gamma = \gamma}$ for each $n$\;
\textbf{Output:}\;
$Y \sim P_{Y \mid \Gamma, X = x}$ and original message stream $s$\;
\;
\tcp{Stage 1: recover $K$ and hence $Z_K \sim P_{Y \mid \Gamma}$}
$(s, K) \gets \ANSdecode(s \mid \zeta)$\;
$\hat{Y} \gets Z_{K}$\;
\;
\tcp{Stage 2: compute $\Gamma$ using the singularity of the channel; see \Cref{eq:singularity_condition}.}
$\Gamma \gets Q_\Delta(\log g(\hat{Y}))$\;
\;
\tcp{Stage 3: recover $N$ and hence $\Upsilon_N \sim P_{Y \mid X, \Gamma}$}
$M(\Gamma) \gets 2^{\Gamma + \Delta}$\;
$N \gets \ANSdecode(s \mid \GeomDist(M(\Gamma))$\;
\;
\tcp{Stage 4: re-encode the GRS decision bits into the stream}
\For{$k = K, K-1,\dotsc,1$}{
\tcp{Compute acceptance probability}
$A_k \gets \min\left\{1, \max\left\{0, \left(\frac{dP_{Y \mid \Gamma}}{dP_Y}(Z_{k}) - L_k\right) \Big/ S_k \right\}\right\}$ \;
\;
\tcp{We know the $K$th decision was an accept, so re-encode a $1$.
All decisions before the $K$th one were rejects, so re-encode a $0$.}
$s \gets \ANSencode(s, 1 \textbf{ if } k = K \textbf{ else } 0 \mid \Bernoulli(A_k))$
}
\tcp{$\Upsilon_N$ is selected from the common randomness $\{\Upsilon_n\}$ where $\Upsilon_n \sim P_{Y \mid \Gamma}$}
\Return{$\Upsilon_N, s$}
\caption{Bits-back rejection sampling decoder}
\label{alg:bbgrs_decoder}
\end{algorithm}
\end{minipage}%
\end{figure}%

\end{document}